\PassOptionsToPackage{unicode}{hyperref}
\PassOptionsToPackage{hyphens}{url}
\PassOptionsToPackage{dvipsnames,svgnames,x11names}{xcolor}
\documentclass[
  authoryear,
  preprint]{elsarticle}
\usepackage{xcolor}
\usepackage{amsmath,amssymb}
\usepackage{iftex}
\ifPDFTeX
  \usepackage[T1]{fontenc}
  \usepackage[utf8]{inputenc}
  \usepackage{textcomp} % provide euro and other symbols
\else % if luatex or xetex
  \usepackage{unicode-math} % this also loads fontspec
  \defaultfontfeatures{Scale=MatchLowercase}
  \defaultfontfeatures[\rmfamily]{Ligatures=TeX,Scale=1}
\fi
\usepackage{lmodern}
\ifPDFTeX\else
\fi
\IfFileExists{upquote.sty}{\usepackage{upquote}}{}
\IfFileExists{microtype.sty}{% use microtype if available
  \usepackage[]{microtype}
  \UseMicrotypeSet[protrusion]{basicmath} % disable protrusion for tt fonts
}{}
\makeatletter
\@ifundefined{KOMAClassName}{% if non-KOMA class
  \IfFileExists{parskip.sty}{%
    \usepackage{parskip}
  }{% else
    \setlength{\parindent}{0pt}
    \setlength{\parskip}{6pt plus 2pt minus 1pt}}
}{% if KOMA class
  \KOMAoptions{parskip=half}}
\makeatother
\makeatletter
\ifx\paragraph\undefined\else
  \let\oldparagraph\paragraph
  \renewcommand{\paragraph}{
    \@ifstar
      \xxxParagraphStar
      \xxxParagraphNoStar
  }
  \newcommand{\xxxParagraphStar}[1]{\oldparagraph*{#1}\mbox{}}
  \newcommand{\xxxParagraphNoStar}[1]{\oldparagraph{#1}\mbox{}}
\fi
\ifx\subparagraph\undefined\else
  \let\oldsubparagraph\subparagraph
  \renewcommand{\subparagraph}{
    \@ifstar
      \xxxSubParagraphStar
      \xxxSubParagraphNoStar
  }
  \newcommand{\xxxSubParagraphStar}[1]{\oldsubparagraph*{#1}\mbox{}}
  \newcommand{\xxxSubParagraphNoStar}[1]{\oldsubparagraph{#1}\mbox{}}
\fi
\makeatother

\usepackage{longtable,booktabs,array}
\usepackage{calc} % for calculating minipage widths
\usepackage{etoolbox}
\makeatletter
\patchcmd\longtable{\par}{\if@noskipsec\mbox{}\fi\par}{}{}
\makeatother
\IfFileExists{footnotehyper.sty}{\usepackage{footnotehyper}}{\usepackage{footnote}}
\makesavenoteenv{longtable}
\usepackage{graphicx}
\makeatletter
\newsavebox\pandoc@box
\newcommand*\pandocbounded[1]{% scales image to fit in text height/width
  \sbox\pandoc@box{#1}%
  \Gscale@div\@tempa{\textheight}{\dimexpr\ht\pandoc@box+\dp\pandoc@box\relax}%
  \Gscale@div\@tempb{\linewidth}{\wd\pandoc@box}%
  \ifdim\@tempb\p@<\@tempa\p@\let\@tempa\@tempb\fi% select the smaller of both
  \ifdim\@tempa\p@<\p@\scalebox{\@tempa}{\usebox\pandoc@box}%
  \else\usebox{\pandoc@box}%
  \fi%
}
\def\fps@figure{htbp}
\makeatother

\usepackage[]{natbib}
\usepackage{booktabs}
\usepackage{longtable}
\usepackage{array}
\usepackage{multirow}
\usepackage{wrapfig}
\usepackage{float}
\usepackage{colortbl}
\usepackage{pdflscape}
\usepackage{tabu}
\usepackage{threeparttable}
\usepackage{threeparttablex}
\usepackage[normalem]{ulem}
\usepackage{makecell}
\usepackage{xcolor}
\usepackage{booktabs} \usepackage{longtable} \usepackage{placeins} \usepackage{setspace}  \usepackage{hyperref} \hypersetup{ colorlinks=true, linkcolor=blue, citecolor=blue, urlcolor=blue } \usepackage{amsthm} \newtheorem{assumption}{Assumption} \newtheorem{proposition}{Proposition} \newtheorem{definition}{Definition}  \newtheorem{corollary}{Corollary} \newpageafter{abstract}
\makeatletter
\@ifpackageloaded{caption}{}{\usepackage{caption}}
\AtBeginDocument{%
\ifdefined\contentsname
  \renewcommand*\contentsname{Table of contents}
\else
  \newcommand\contentsname{Table of contents}
\fi
\ifdefined\listfigurename
  \renewcommand*\listfigurename{List of Figures}
\else
  \newcommand\listfigurename{List of Figures}
\fi
\ifdefined\listtablename
  \renewcommand*\listtablename{List of Tables}
\else
  \newcommand\listtablename{List of Tables}
\fi
\ifdefined\figurename
  \renewcommand*\figurename{Figure}
\else
  \newcommand\figurename{Figure}
\fi
\ifdefined\tablename
  \renewcommand*\tablename{Table}
\else
  \newcommand\tablename{Table}
\fi
}
\@ifpackageloaded{float}{}{\usepackage{float}}
\floatstyle{ruled}
\@ifundefined{c@chapter}{\newfloat{codelisting}{h}{lop}}{\newfloat{codelisting}{h}{lop}[chapter]}
\floatname{codelisting}{Listing}

\makeatother
\makeatletter
\@ifpackageloaded{caption}{}{\usepackage{caption}}
\@ifpackageloaded{subcaption}{}{\usepackage{subcaption}}
\makeatother
\journal{Economics of Education Review}
\usepackage{bookmark}
\IfFileExists{xurl.sty}{\usepackage{xurl}}{} % add URL line breaks if available
\hypersetup{
  pdftitle={Partial Identification of Causal Effects that Vary by Setting},
  pdfauthor={Nick Huntington-Klein},
  colorlinks=true,
  linkcolor={blue},
  filecolor={Maroon},
  citecolor={Blue},
  urlcolor={Blue},
  pdfcreator={LaTeX via pandoc}}

\begin{document}

\begin{frontmatter}
\title{Partial Identification of Causal Effects that Vary by Setting}
\author[1]{Nick Huntington-Klein%
}
 \ead{nhuntington-klein@seattleu.edu} 

\affiliation[1]{organization={Seattle University, Department of
Economics},city={Seattle},country={United States of
America},countrysep={,},postcodesep={}}

\cortext[cor1]{Corresponding author}

\begin{abstract}
The estimation of causal effects using quasiexperiments often relies on
the use of unusual or serendipitous sources of exogenous variation. When
the goal is estimating the same causal effects across many different
settings, the same unusual exogenous variation often does not exist in
all settings, and the only available form of identification is
selection-on-observables, which relies on a conditional independence
assumption. Partial identification is especially valuable in this
context, as it allows conditional independence to not hold perfectly.
This paper proposes a method that sharpens the jointly identified set of
causal effects across many settings by making use of unobserved
relationships between omitted variable biases across settings.
\end{abstract}

\end{frontmatter}

\section{Introduction}\label{introduction}

In the applied social sciences, researchers often want to estimate and
report ``the effect'' of a given policy or treatment, with a causal
interpretation. ``The effect'' is a knowingly fictional estimand, as it
implies that a given policy or treatment has only one effect.
Researchers generally acknowledge that a given policy is likely to have
many different effects depending on location or the population being
treated \citep{huntington2021effect}. As a result, there is a large and
growing literature on ways to estimate heterogeneous treatment effects,
or ``conditional average treatment effects'' (CATE), which are effects
that are allowed to vary across the population based on observed
criteria like setting or demographics \citep[e.g.,][among many
others]{wager2018estimation}.

However, there is an inherent limitation in the application of
heterogeneous treatment effects in the context of causal inference. In
many cases, we estimate causal effects using quasiexperiments that rely
on some source of exogenous variation in a treatment. Estimating a
heterogeneous effect for a quasiexperiment requires that that same
source of exogenous variation applies in each of the different settings
where we estimate an effect. Since many sources of exogenous variation
are bespoke, unusual, or serendipitous, this is unlikely. Identification
using a conditional independence assumption, if it is plausible, can be
applied in a much wider range of settings than a design relying on an
instrumental variable, difference-in-differences design, or similar. As
a result, methods for estimating heterogeneous treatment effects often
focus on cases where identification relies on selection-on-observables
(the use of covariates / control variables with no explicit source of
exogenous variation).

In many fields, especially in economics, we are skeptical that causal
effects can be identified using selection on observables. However, if we
doubt the conditional independence assumption (or, similarly,
\(d\)-separation) assumption necessary to identify a causal effect using
control variables alone, we can quantify that doubt and provide a set of
partial identification (or ``set identification'') bounds that allow for
conditional independence to not hold exactly. Partial identification
bounds for the case of residual omitted variable bias are often
calculated as a form of sensitivity bounds, where estimates are
recalculated under a bounded range of plausible assumption violations
\citep{cinelli2019, oster2019unobservable}.

In this paper, I provide a method for bridging the gap between, on one
hand, heterogeneous causal effect estimation, and on the other hand
partial identification when the conditional independence assumption does
not hold exactly.

Methods like \citet{cinelli2019} can already produce partial
identification bounds in many different settings independently. However,
estimating partial identification bounds separately by setting ignores
the fact that there are likely to be dependencies between the assumption
violations in different settings. For example, if omitted variable bias
in setting \(j\) is positive, then this might make us more likely to
believe that it is positive in setting \(k\) as well. This paper uses
those dependencies to add additional information. In some cases this
additional information can shrink partial identification bounds on
individual effects, but more often the additional information shrinks
bounds on the partially identified \emph{joint} distribution of effects,
producing a more informative result. The joint distribution of effects
can then be explored by ``pinning'' estimates - forcing an effect in
setting \(j\) to some part of its identified set, such as the upper or
lower end of that distribution, and seeing what effects in other
settings are consistent with that estimate.

The method is fairly general and requires only a set of partial
identification bounds for each setting, which can be produced by any
number of different methods, like \citet{cinelli2019}. I add to this a
matrix of bounds that tie together the residual omitted variable bias
across settings. This matrix can be proposed by assumption, but I also
propose a method of ``supershort'' regression that can produce plausible
bounds. This matrix can be used to narrow the provided set of partial
identification bounds, and can also answer questions like ``if the
effect in setting \(j\) is near the upper end of its range, what is the
identified set of estimates for setting \(k\)?''

In this paper I briefly review the literature on partial identification
and heterogeneous treatment effect estimation, then introduce the
method. I present software packages able to implement the method in R,
Python, and Stata, and then demonstrate the method using an example
where I estimate the labor market returns to attending college in the
United States and see how the return differs over time and across
states.

\section{Partial Identification and Heterogeneous Treatment
Effects}\label{partial-identification-and-heterogeneous-treatment-effects}

Researchers have long been aware of the fact that our typical approach
to performing causal inference relies on assumptions that are necessary
to identify an effect but are unlikely to be true (``strong''
assumptions). Two potential reactions to this fact are either to ignore
the problem, or to conclude that this makes causal inference impossible.
A third approach is partial identification, also known as set
identification, in which the author only makes assumptions that they
actually believe are plausible (``weak'' assumptions), and then asks
which estimates are consistent with those assumptions, accepting that
what is identified may be a range of answers rather than a single answer
\citep{manski2020lure}.

Partial identification is a general approach. The procedure for
constructing identified sets differs depending on the analysis being
performed as well as which assumptions are being weakened, and there are
many approaches to doing so \citep{tamer2010partial}. In this paper we
focus specifically on the case of omitted variable bias in linear
models. In a standard case, a researcher wishes to identify the causal
effect of \(X\) on \(Y\) by conditioning on a set of control variables
\(Z\):

\[Y=\beta_0+\beta_1X+\beta_2Z+\varepsilon\]

Concluding that a regression of \(Y\) on \(X\) and \(Z\) produces a
\(\hat{\beta}_1\) that identifies the causal effect of \(X\) on \(Y\)
requires the assumption \(Cov(X,\varepsilon)=0\). We can refer to this
as a conditional independence assumption, or think of the closely
related concept of \(d\)-separation in the field of structural causal
modeling \citep{pearl2018book}.

Current approaches to partial identification concerning the conditional
independence assumption often take the form of imposing weaker
restrictions on \(Cov(X,\varepsilon)\), for instance assuming that
\(Cov(X,\varepsilon)\in[\mu_L,\mu_U]\) instead of requiring
\(Cov(X,\varepsilon)=0\). Any non-0 value implies that \(\hat{\beta}_1\)
is a biased estimate of \(\beta_1\), but the
\(Cov(X,\varepsilon)\in[\mu_L,\mu_U]\) bounds limit the potential size
of the bias and thus imply a set of bounds on \(\beta_1\).

Allowing \(Cov(X,\varepsilon)\neq 0\) and calculating the implied
\(\beta_1\) values that result from the weakening of this assumption is
also referred to as sensitivity analysis. Results follow directly and
clearly from the standard formula for omitted variable bias, and the
only remaining task is to determine plausible \([\mu_L,\mu_U]\) bounds,
which often proceeds by looking at the amount of bias removed by the
inclusion of measured controls.

There is a long history of this approach, dating back at least as far as
\citet{cornfield1959smoking}. This sensitivity analysis process has seen
renewed interest over the past few decades, spurred by
\citet{imbens2003sensitivity} and \citet{altonji2005selection} and with
work following through to the current day with
\citet{oster2019unobservable} and a nonparametric variant in
\citet{masten2018identification}. There is a long related literature
taking a similar approach in the domain of matching, which generally
follows from \citet{rosenbaum1983assessing} and for which there are many
modern variants
\citep[e.g.,][]{liu2013introduction, yadlowsky2022bounds}.

\citet{cinelli2019} produced a general and updated framework for
sensitivity analysis using partial \(R^2\) values to bound omitted
variable bias and produce plausibility bounds for estimated parameters.
\citet{chernozhukov2026} extends that framework in application to
debiased machine learning.

This paper links the work on partial identification in cases of omitted
variable bias with work on heterogeneous treatment effects. There is
considerable interest in being able to estimate treatment effects that
are heterogeneous across many different settings, especially in the
developing field of causal machine learning \citep[e.g.,][among many
others]{wager2018estimation}. However, identification in these methods
often relies on a conditional independence assumption.

There is already work on partial identification of heterogeneous
effects. In some cases this is a straightforward extension - the
aforementioned methods that perform sensitivity analysis for average
treatment effects can simply be applied in different settings to produce
bounds on conditional average treatment effects (CATE). In some cases
the link between the ATE and CATE bounds are made explicit, as in
\citet{yadlowsky2022bounds} or \citet{lanners2025data} where bounds on
the CATE inform the sensitivity analsysis for the ATE.

Other methods more explicitly focus on heterogeneous treatment effects
(CATE or otherwise) and their partial identification bounds.
\citet{fan2010sharp} use Fréchet--Hoeffding bounds, which limit the
differences between two distributions (in this case, observed and
counterfactual) to produce sharp partial identification bounds on the
distribution of treatment effects. The method is intended for use in a
randomized experiment context, but allows adjustment for covariates.

\citet{chen2023differential} look at cases where two treatment variables
are available. They use a differential-effect approach and show that
effects can be bounded between a differential-effects estimate and an
estimate generated by inverse probability weighting. They specify
partial identification bounds as well as the conditions under which this
allows for both average treatment effects and conditional average
treatment effects to be point identified.

The closest structural analogue to the present paper is
\citet{rambachan2023more}, which imposes restrictions on the change in
parallel-trends violations across periods in difference-in-differences
contexts. Like in this paper, \citet{rambachan2023more} bounds the
degree to which effects vary across settings. I pursue an analogous
agenda for a general omitted variable bias setting.

This paper expands upon previous work in partial identification of
conditional average treatment effects, in particular the sensitivity
analysis approach.

\section{Partial Identification of Effect
Variation}\label{partial-identification-of-effect-variation}

\subsection{Notation}\label{notation}

Let \(\mathcal{S} = \{1,\ldots,K\}\) denote a finite set of distinct
\emph{settings,} such as geographic regions, time periods, or subgroups.

\begin{assumption}[Partially Linear Model]
\label{asm:model}
In each setting $i \in \mathcal{S}$, the data-generating process satisfies the partially linear model
\[
Y^i = \theta^i D^i + f^i(X^i, A^i) + \varepsilon^i,
\]
where $Y^i$ is a continuously-valued outcome, $D^i$ is a real-valued treatment, $X^i$ is a vector of observed covariates, $A^i$ is a vector of unobserved covariates, and $E[\varepsilon^i \mid D^i, X^i, A^i] = 0$. The object of interest is the setting-specific causal effect $\theta^i$.
\end{assumption}

Since \(A^i\) is unobserved, the researcher estimates \(\theta^i\) using
the ``short'' regression that omits \(A^i\):

\[Y^i = \theta^i_s D^i + f^i_s(X^i) + \varepsilon^i_s.\]

Following \citet{chernozhukov2026}, we can express the ``long''
parameter \(\theta^i\) via a Frisch-Waugh-Lovell decomposition as the
projection of \(Y^i\) on the residuals from predicting \(D^i\) using
\(A^i\) and \(X^i\):

\[\theta^i = E\left[Y^i \frac{D^i - E[D^i|X^i,A^i]}{E\left[(D^i - E[D^i|X^i,A^i])^2\right]}\right] \equiv E[Y^i \alpha^i]\]

and similarly \(\theta^i_s \equiv E[Y^i \alpha^i_s]\), where
\(\alpha^i_s\) omits \(A^i\) in the denominator projection. \(g^i\)
represents the outcome regression function such that
\(g^i = E[Y^i|D^i, X^i, A^i]\) and \(g^i_s = E[Y^i|D^i,X^i]\). Then, the
omitted variable bias (OVB) incurred by using the short regression is:

\[B^i \equiv \theta^i_s - \theta^i = E[(g^i_s - g^i)(\alpha^i_s - \alpha^i)]\]

This is effectively the correlation between regression error
(\(g^i_s - g^i\)) and error in prediction of the treatment
(\(\alpha^i_s - \alpha^i\)). From here I begin to deviate from
\citet{chernozhukov2026}.

From the definition of bias, we can assume a bound on the size of that
bias, which immediately implies an identified set for \(\theta_s^i\):

\begin{assumption}[OVB Boundedness]
\label{asm:individual-bounds}
For each setting $i \in \mathcal{S}$, the omitted variable bias satisfies
\[
B^i \in [\nu^i_l, \nu^i_u]
\]
for known constants $\nu^i_l \leq \nu^i_u$.
\end{assumption}

Section \ref{sec-selecting-baseline-plausible-bounds} discusses how to
choose \([\nu^i_l, \nu^i_u]\) in practice. The implied bounds on
\(\theta^i\) follow immediately from the definition
\(\theta^i = \theta^i_s - B^i\).

\begin{proposition}[The Identified Set for an Individual Setting]
\label{prop:individual}
Under Assumptions \ref{asm:model} and \ref{asm:individual-bounds}, the true causal parameter satisfies
\[
\theta^i \in \mathcal{I}^i \equiv [\theta^i_s - \nu^i_u, \theta^i_s - \nu^i_l].
\]
\end{proposition}

The individual bounds \(\mathcal{I}^i\) treat each setting
independently. The remainder of the paper exploits cross-setting
dependence to obtain additional restrictions.

\subsection{Omitted Variable Bias in Multiple
Settings}\label{sec-multiple-settings}

Consider two settings \(j, k \in \mathcal{S}\). Two short regressions
estimated separately in each of these settings would be affected by
omitted variable biases (OVB) in those settings \(B^j\) and \(B^k\).
Assumptions bounding those biases would lead to two identified sets
\(\theta^j\) and \(\theta^k\). If we assume that those biases are
related in some way, then the joinly identified set
\([\theta^j,\theta^k]\) can have narrower bounds than each of the
setting-specific identified sets alone.

Here consider a case where a proportionality assumption applies, whereby
the bias in one setting is an unknown proportion of the bias in another
setting:

\begin{assumption}[Proportional Bias Relationship]
\label{asm:proportional}
For settings $j,k \in \mathcal{S}$, the omitted variable biases satisfy
\[
B^j = \rho^{jk} B^k,
\]
where $\rho^{jk} \in [\rho^{jk}_l,\; \rho^{jk}_u]$ for known constants $0<\rho^{jk}_l \leq 1$ and $\rho^{jk}_u \geq 1$.
\end{assumption}

\textbf{Structural motivation.} A natural data-generating process that
produces Assumption \ref{asm:proportional} is a single-factor model for
unobserved confounding. \(A^i\) is made up of a single common unobserved
factor (or a composite factor representing multiple variables). In this
case, the omitted variable bias in each setting is a scalar: the product
of the relationship between \(D^i\) and \(A^i\), and the direct impact
of \(A^i\) on \(Y^i\). The ratio of these scalars is itself scalar and
so there is a proportional relationship, as long as the signs match.
Outside of the univariate context, the partial-\(R^2\) parameterization
for OLS bias from \citet{cinelli2019} can also justify a proportionality
relationship between the two biases. Given the proportionality
assumption, bounding \(\rho^{jk}\) is then equivalent to bounding the
relative strength of unobserved confounding across settings.

\textbf{Implicit sign restriction.} Assumption \ref{asm:proportional}
with \(\rho^{jk} \in [\rho^{jk}_l, \rho^{jk}_u]\) and
\(\rho^{jk}_u\geq\rho^{jk}_l \geq 0\) implies that \(B^j\) and \(B^k\)
have the \emph{same sign} (or at least one is zero). This is a
substantive restriction. The direction of selection can plausibly differ
across contexts. Researchers should assess whether same-sign bias is
plausible before applying Assumption \ref{asm:proportional}. If sign
reversals are plausible for a given \(j\) and \(k\), then the researcher
may not wish to consider bias in one setting as informative of bias in
another, and would not want to apply Assumption \ref{asm:proportional}
to that pair.

\textbf{The zero-bias pathology.} If \(B^k = 0\) exactly, then
Assumption \ref{asm:proportional} forces \(B^j = 0\) regardless of
\(\rho^{jk}_u\). In practice, a near-zero \emph{estimated} bias in one
setting may reflect sampling noise rather than genuine absence of
confounding. Researchers should be cautious: if a setting's bias is near
zero, cross-setting restrictions propagated through that setting will be
artificially tight.

The proportionality assumption allows for the bounds on the effect in
each setting to shrink the jointly identified bounds:

\begin{definition}[Bias Difference Set]
\label{def:C}
Under Assumptions \ref{asm:individual-bounds} and \ref{asm:proportional}, define $D^{jk} = B^j - B^k = (\rho^{jk} - 1)B^k$. Since $D^{jk}$ is the product of $\rho^{jk} \in [\rho^{jk}_l, \rho^{jk}_u]$ and $B^k \in [\nu^k_l, \nu^k_u]$, the bounds on $D^{jk}$ come from a set of candidate extremes $\mathbb{C}^{jk}$:
\[
\mathbb{C}^{jk} = \left\{(\rho^{jk}_l-1)\nu^k_l, (\rho^{jk}_u-1)\nu^k_l, (\rho^{jk}_l-1)\nu^k_u, (\rho^{jk}_u-1)\nu^k_u\right\},
\]
and let $c^{jk}_l = \min(\mathbb{C}^{jk})$ and $c^{jk}_u = \max(\mathbb{C}^{jk})$. Then $D^{jk} \in [c^{jk}_l, c^{jk}_u]$.
\end{definition}

The value of \(\rho^{jk}\) as well as the signs of \(\nu^k_l\) and
\(\nu^k_u\) determine which of the four elements of \(\mathbb{C}^{jk}\)
is the minimum or maximum, so we do not simplify further in general.

Depending on context, it may be justifiable to add an additional
restriction on the proportionality assumption that the bounds are
proportionally symmetrical such that \(\rho^{jk}_l=1/\rho^{jk}_u\). In
that case the bounds on \(D^{jk}\) can be simplified to be dependent on
three parameters rather than four:

\begin{definition}[Bias Difference Set Under Symmetry]
\label{def:Csymmetry}
Under Assumptions \ref{asm:individual-bounds} and \ref{asm:proportional}, define $D^{jk} = B^j - B^k = (\rho^{jk} - 1)B^k$. Since $D^{jk}$ is the product of $\rho^{jk} \in [1/\rho^{jk}_u, \rho^{jk}_u]$ and $B^k \in [\nu^k_l, \nu^k_u]$, the bounds on $D^{jk}$ come from a set of candidate extremes $\mathbb{C}^{jk}$:
\[
\mathbb{C}^{jk} = \left\{\frac{1-\rho^{jk}_u}{\rho^{jk}_u}\nu^k_l, (\rho^{jk}_u-1)\nu^k_l, \frac{1-\rho^{jk}_u}{\rho^{jk}_u}\nu^k_u, (\rho^{jk}_u-1)\nu^k_u\right\},
\]
and let $c^{jk}_l = \min(\mathbb{C}^{jk})$ and $c^{jk}_u = \max(\mathbb{C}^{jk})$. Then $D^{jk} \in [c^{jk}_l, c^{jk}_u]$.
\end{definition}

\subsection{The Jointly Identified
Set}\label{sec-jointly-identified-set}

Using Definition \ref{def:C} (or, alternately, Definition
\ref{def:Csymmetry}), we can now state the main result for two settings:
the joint identified set for the pair \((\theta^j, \theta^k)\).

\begin{proposition}[Joint Identified Set for Two Settings]
\label{prop:joint}
Under Assumptions \ref{asm:model}, \ref{asm:individual-bounds}, and \ref{asm:proportional}, the true pair $(\theta^j, \theta^k)$ belongs to the joint identified set
\[
\mathcal{J}^{jk} = \left\{(\theta^j, \theta^k) : \theta^j \in \mathcal{I}^j, \theta^k \in \mathcal{I}^k, \theta^j - \theta^k \in \mathcal{I}_D^{jk}\right\}
\]
where
\[
\mathcal{I}_D^{jk} = \left[(\theta^j_s - \theta^k_s) - c^{jk}_u, (\theta^j_s - \theta^k_s) - c^{jk}_l\right].
\]
\end{proposition}

\begin{proof}
By Proposition \ref{prop:individual}, $\theta^j \in \mathcal{I}^j$ and $\theta^k \in \mathcal{I}^k$. For the difference constraint, note that
\[
\theta^j - \theta^k = (\theta^j_s - B^j) - (\theta^k_s - B^k) = (\theta^j_s - \theta^k_s) - (B^j - B^k) = (\theta^j_s - \theta^k_s) - D^{jk}.
\]
Since $D^{jk} \in [c^{jk}_l,\; c^{jk}_u]$ by Definition \ref{def:C}, it follows that
\[
\theta^j - \theta^k \in \left[(\theta^j_s - \theta^k_s) - c^{jk}_u, (\theta^j_s - \theta^k_s) - c^{jk}_l\right] = \mathcal{I}_D^{jk},
\]
which completes the proof.
\end{proof}

\textbf{Remark on joint versus marginal sets.} Proposition
\ref{prop:joint} establishes that cross-setting restrictions constrain
the joint identified set \(\mathcal{J}^{jk}\). This may or may not
constrain the marginal identified sets for \(\theta^j\) and
\(\theta^k\). Effectively, any value in \(\mathcal{I}^j\) that is
consistent with \emph{any} value of \(\theta^j\in\mathcal{I}^k\) will
remain in the marginal identified set for \(\theta^j\), while only
values that are inconsistent with \emph{all} values of
\(\theta^j\in\mathcal{I}^k\) will be omitted from the marginal
identified set. So while the joint identified set adds information that
reduces the jointly identified parameter space, shrinkage of the
marginal identfied sets may occur sometimes but not always. The
following proposition characterizes when marginal shrinkage occurs.

\begin{proposition}[Marginal Sharpening Conditions]
\label{prop:marginal}
The marginal identified set for $\theta^j$ derived from the joint set $\mathcal{J}^{jk}$ is
\[
\mathrm{proj}_j(\mathcal{J}^{jk}) = \mathcal{I}^j \cap \left[\theta^j_s - \nu^k_u - c^{jk}_u, \theta^j_s - \nu^k_l - c^{jk}_l\right].
\]
Strict marginal sharpening occurs when $\mathrm{proj}_j(\mathcal{J}^{jk}) \subsetneq \mathcal{I}^j$. Marginal sharpening happens if and only if at least one of the following holds:
\begin{enumerate}
\item $\nu^k_u + c^{jk}_u < \nu^j_u$ (the lower bound of $\theta^j$ is raised);
\item $\nu^k_l + c^{jk}_l > \nu^j_l$ (the upper bound of $\theta^j$ is lowered).
\end{enumerate}
\end{proposition}

Proof is in Appendix Section~\ref{sec-apx-marginal}.

In addition to showing where the jointly identified set places
restrictions on the marginal sets, the marginal sharpening conditions
also show when the joint restrictions make a marginal set empty such
that no parameter value is consistent with the assumed restrictions. In
these cases, the bounds imposed on at least one of \(\theta^j\),
\(\theta^k\), or \(\rho^{jk}\) are inconsistent with the data and must
be relaxed.

\textbf{Worked example and the value of the joint restriction.} To see
the implications of Proposition \ref{prop:joint} and Corollary
\ref{cor:symmetric} in \ref{sec-apx-marginal}, consider
\(\theta^j_s = \theta^k_s = 1\), \(\nu^j_l = \nu^k_l = -1\),
\(\nu^j_u = \nu^k_u = 1\), and \(\rho^{jk}_u = 2\). This is a symmetric
equal-bounds case. By Corollary \ref{cor:symmetric} in
\ref{sec-apx-marginal}, no marginal sharpening occurs:
\(\text{proj}_j(\mathcal{J}^{jk}) = \mathcal{I}^j = [0, 2]\).

The elements of \(\mathbb{C}^{jk}\) are \(\{1/2, -1, -1/2, 1\}\), giving
\(c^{jk}_l = -1\) and \(c^{jk}_u = 1\), so the difference constraint is
\(\mathcal{I}^D_{jk} = [0-1, 0+1] = [-1, 1]\). The full restrictions are
\(\theta^j \in [0,2]\), \(\theta^k \in [0,2]\), and
\(\theta^j - \theta^k \in [-1,1]\). The third constraint rules out
combinations like \((\theta^j, \theta^k) = (0, 2)\) or \((2, 0)\), which
would have been permitted by the individual bounds alone. This is the
value of the joint restriction: it restricts the set of admissible
combinations, even when it does not shrink the individual ranges.

\subsection{\texorpdfstring{Extension to \(K > 2\)
Settings}{Extension to K \textgreater{} 2 Settings}}\label{sec-k-settings}

With \(K > 2\) settings, the proportionality assumption (Assumption
\ref{asm:proportional}) is applied for each pair \((j, k)\) where a
restriction is desired. A researcher may choose to impose no constraint
on some pairs \((j,k)\) where bias in one setting is presumed not to be
informative about bias in the other by setting \(\rho^{jk}_l=-\infty\)
and \(\rho^{jk}_u=\infty\) (\(\rho^{jk}_l<0\) violates Assumption
\ref{asm:proportional} for that pair since their biases are not assumed
to have a proportional relationship).

\begin{proposition}[$K$-Setting Joint Identified Set]
\label{prop:k-settings}
Under Assumption \ref{asm:model} and Assumption \ref{asm:individual-bounds} for all $i \in \mathcal{S}$, and Assumption \ref{asm:proportional} for all pairs $(j,k)$, the true vector $(\theta^1, \ldots, \theta^K)$ belongs to the polytope
\[
\mathcal{J} = \left\{(\theta^1,\ldots,\theta^K) : \theta^i \in \mathcal{I}^i \forall i, \theta^j - \theta^k \in \mathcal{I}^D_{jk} \forall (j,k)\right\}.
\]
This polytope is defined by at most $2K + K(K-1)$ linear inequality constraints in $K$ unknowns. The marginal identified set for any $\theta^i$ is the interval $[\underline{\theta}^i, \overline{\theta}^i]$ where
\[
\underline{\theta}^i = \min_{(\theta^1,\ldots,\theta^K)\in\mathcal{J}} \theta^i, \qquad \overline{\theta}^i = \max_{(\theta^1,\ldots,\theta^K)\in\mathcal{J}} \theta^i,
\]
each of which is the solution to a linear program.
\end{proposition}

\begin{proof}
By Proposition \ref{prop:joint} applied to each pair, the true vector belongs to $\mathcal{J}^{jk}$ for all relevant pairs. The intersection of all these pairwise joint sets and the individual bounds gives $\mathcal{J}$. Each individual bound $\theta^i \in [\theta^i_s - \nu^i_u, \theta^i_s - \nu^i_l]$ contributes 2 linear constraints; each difference bound $\theta^j - \theta^k \in [d^{jk}_l, d^{jk}_u]$ contributes 2 linear constraints for each ordered pair. With $K(K-1)/2$ pairs, each contributing 2 constraints, plus $2K$ individual constraints, the total number of constraints is $2K + K(K-1)$, some of which may be trivial constraints for pairs with infinitely wide bounds on their differences.
\end{proof}

\textbf{Remark on pairwise consistency.} With \(K\) settings, the
\(K(K-1)/2\) pairwise proportionality assumptions are not mutually
independent. Transitivity constraints apply: if \(B^j = \rho^{jk} B^k\)
and \(B^k = \rho^{km} B^m\), then \(B^j = \rho^{jk}\rho^{km} B^m\),
implying \(\rho^{jm} = \rho^{jk}\rho^{km}\). If the three pairwise
bounds \(\rho^{jk}_u\), \(\rho^{km}_u\), \(\rho^{jm}_u\) are estimated
separately such that their values are not jointly determined, there is
no guarantee that they satisfy this transitivity. Inconsistent pairwise
bounds may cause the joint polytope \(\mathcal{J}\) to be small or
perhaps null, because the difference constraints impose implicit
cross-constraints.

\textbf{Remark on statistical inference.} The identified sets
\(\mathcal{J}\) and its projections are population objects: they
characterize the range of \((\theta^1,\ldots,\theta^K)\) consistent with
the stated assumptions, treating the short-regression estimates
\(\theta^i_s\) and the bias bounds as known. In practice,
\(\hat\theta^i_s\) is estimated with sampling error, and the estimated
bias bounds also carry uncertainty. A complete inference framework must
address both identification uncertainty (the width of the identified
set) and sampling uncertainty (variability of \(\hat\theta^i_s\) and
\(\hat\nu\)). \citet{imbens2004confidence} and \citet{stoye2009more}
provide frameworks for constructing confidence regions for partially
identified parameters in the single-setting case; extending these
frameworks to the multi-setting joint polytope \(\mathcal{J}\) is a
natural direction for future work. In the current paper, the reported
bounds should be interpreted as population identified sets under the
stated assumptions, rather than as confidence sets that account for
sampling error.

\section{Details in Application}\label{sec-application}

\subsection{Exploring the Joint Set via ``Pinning''}\label{sec-pinning}

The addition of joint restrictions on bias is not guaranteed to shrink
marginal identified sets. The primary contribution of this paper is
therefore the restriction on the \emph{joint} identified set
\(\mathcal{J}^{jk}\), which restricts the heterogeneity in effects by
bounding the difference \(\theta^j - \theta^k\). The resulting
identified set is then an \(N\)-dimensional object, where \(N\) is the
number of settings. Reporting an estimated \(N\)-dimensional object is
qualitatively difficult, and conclusions may be difficult to draw.

A meaningful way to explore the resulting identified set is by
``pinning'' one of the settings, where a single value is chosen for one
of the settings \(j\) (the estimated effect in setting \(j\) is
``pinned'' to a specific value). Any restrictions between setting \(j\)
and other settings then sharpens the jointly identified set, showing
plausible values of estimates in other settings given the initial pinned
value. Then, by exploring the jointly identified set given different
pinned values from different parts of the marginal identified set of
\(j\), the researcher can demonstrate the joint dependencies between
settings and make claims like ``conditional on the effect in setting
\(j\) being near the (upper end / lower end / middle) of its marginal
identified set, then here are the marginal distributions of effects in
other settings consistent with that value.'' More colloquially, ``at the
low end, the identified sets look like this, and at the high end, they
look like this.''

\begin{definition}[Conditional Identified Set / Pinning]
\label{def:pinning}
Given the joint identified set $\mathcal{J}^{jk}$ from Proposition \ref{prop:joint} and a fixed value $t \in \mathcal{I}^j$, the \emph{conditional identified set} for $\theta^k$ given $\theta^j = t$ is
\[
\mathcal{I}^k(t) = \mathcal{I}^k \cap \left[t - (\theta^j_s - \theta^k_s) + c^{jk}_l, t - (\theta^j_s - \theta^k_s) + c^{jk}_u\right].
\]
This set is always at most as large as $\mathcal{I}^k$ and may be strictly smaller whenever the cross-setting constraint is active.
\end{definition}

\emph{Derivation.} Given \(\theta^j = t\), the constraint
\(\theta^j - \theta^k \in \mathcal{I}^D_{jk} = [d_l, d_u]\) implies
\(\theta^k \in [t - d_u, t - d_l]\). Substituting
\(d_u = (\theta^j_s - \theta^k_s) - c^{jk}_l\) and
\(d_l = (\theta^j_s - \theta^k_s) - c^{jk}_u\) gives the result.

Continuing the worked example from
Section~\ref{sec-jointly-identified-set}: fixing \(t = 0\) (the lower
end of \(\mathcal{I}^j\)), \[
\mathcal{I}^k(0) = [0,2] \cap \left[0 - 0 + (-1),\; 0 - 0 + 1\right] = [0,2] \cap [-1,1] = [0,1].
\] So pinning \(\theta^j\) to its lower bound restricts \(\theta^k\) to
\([0,1]\), halving its range. Symmetrically, fixing \(t = 2\) restricts
\(\theta^k\) to \([1,2]\).

\subsection{Selecting Baseline Plausible
Bounds}\label{sec-selecting-baseline-plausible-bounds}

There are two kinds of plausible bounds that must be selected to perform
the method in Section~\ref{sec-multiple-settings}. The first is the
selection of the individual-setting bounds \([\nu^i_l, \nu^i_u]\) which
bound the degree of omitted variable bias in each setting \(i\). This
can be done in any number of ways, for example using the methods
outlined in \citet{cinelli2019}, \citet{chernozhukov2026}, or even via
matching with a derivative of Rosenbaum bounds.

The second parameter that must be bounded is \(\rho^{jk}\), the omitted
variable bias proportionality coefficient for each pair of settings
\((j, k)\).

\(\rho^{jk}\) can be bounded using external knowledge. For example, if
the settings represent geographic locations, plausibility can be used to
bound \(\rho^{jk}\) for adjacent locations, and then non-adjacent
locations can have no modeled relationship, or the relationship can
decay be widening the bounds on the relationship as a function of
distance.

A data-based approach to bounding (either for all pairs of settings or a
subset) can make use of an approach similar to \citet{cinelli2019} where
we use the bias resolved by a set of observed variables as a benchmark
for bias that might be resolved by the addition of unobserved variables.
We begin with the short regression to be applied in setting \(i\):

\[ Y^i=\theta^i_s D^i+f^i_s(X)+\varepsilon^i_s\]

and further segment \(X^i\) into two covariate sets, \(X^i_1\) and
\(X^i_2\), where \(X^i_2\) is non-empty. Then, we estimate the short
regression in setting \(i\), as well as the super-short regression:

\[Y^i=\theta^i_{ss}D^i+f^i_{ss}(X^i_1)+\varepsilon^i_{ss}\] Then, the
difference between the estimated \(\theta^i_s\) and \(\theta^i_{ss}\) is
the omitted variable bias that is removed by the addition of \(X^i_2\)
to the model, \(B^i_{ss}=\hat{\theta}^i_s-\hat{\theta}^i_{ss}\).

The formal justification for using \(B^i_{ss}\) to bound \(\rho^{jk}\)
rests on the following assumption.

\begin{assumption}[Proportional Confounding Across Settings]
\label{asm:supershort}
For settings $j$ and $k$ for which $\text{sign}(B^{j})=\text{sign}(B^{k})$, the ratio of unobserved OVBs is bounded by the ratio of observed covariate-induced biases:
\[
\rho^{jk} \in \left[\min\{\frac{B^j_{ss}}{B^k_{ss}},\frac{B^k_{ss}}{B^j_{ss}}\},\max\{\frac{B^j_{ss}}{B^k_{ss}},\frac{B^k_{ss}}{B^j_{ss}}\}\right],
\]
\end{assumption}

Assumption \ref{asm:supershort} states that the settings where observed
confounders generate more bias (i.e.~the supershort analysis leads to
large changes in the estimate) are also the settings where unobserved
confounders generate more bias, in the same relative proportion. The
ratio of observed-covariate biases between two settings then acts as a
guide for the ratio of unobserved-covariate biases. This is analogous to
how \citet{cinelli2019} and similar studies use observed confounding to
bound unobserved confounding, but applied across settings rather than
within a single setting. The assumption fails if the degree of observed
confounding is not correlated with the degree of unobserved confounding
across settings.

\textbf{Remark on the supershort estimator's limitations.} The
supershort approach has several important practical limitations.

First, \(B^j_{ss}\) and \(B^k_{ss}\) are estimated with sampling error,
and the ratio of two estimated quantities is unstable when the
denominator is near zero. If a setting's observed covariate-induced bias
change is near zero, the bounds on \(\rho^{jk}\) will be too wide to be
informative. Given that the proportionality assumption
\ref{asm:proportional} likely fails in these contexts anyway, the
supershort method's imposition of very loose bounds is desirable.
Similarly, the supershort method cannot be meaningfully applied when
\(\text{sign}(B^{i})\neq\text{sign}(B^{j})\), but this is a case where
the data provides little information on the relationship between biases
in these scenarios and so failing to impose a restriction in that
context is desirable.

Second, the choice of which variables to include in \(X^i_1\) versus
\(X^i_2\) is consequential: different partitions yield different
supershort estimates and therefore different \(\rho^{jk}_u\) values.
Researchers should report sensitivity to this partitioning choice.

\section{Software}\label{software}

In order to make the application of the method in this paper easier to
execute, I provide software packages capable of producing matrices of
\(\rho^{jk}\) bounds, as well as partial identification bounds by
setting, and the narrowed partial identification bounds that emerge as a
result of applying the \(\rho^{jk}\) restrictions.

The package is called \textbf{hetset}, short for
\emph{het}erogeneous-effect \emph{set} identification, and is available
for the R language on CRAN (install with
\texttt{install.packages(\textquotesingle{}hetset\textquotesingle{})}),
for the Python language on PyPi (install with
\texttt{pip\ install\ hetset}), and for the Stata language on ssc
(install with \texttt{ssc\ install\ hetset}).

In all three packages, the method can be applied by first estimating
partial identification bounds separately in each setting with the
\texttt{ovb\_bounds\_by\_setting} function. By default, this uses the
\textbf{sensemakr} package (also available in all three languages) to
apply \citet{cinelli2019}. Then, the user can specify their own set of
\(\rho^{jk}\) bounds, or can use the function
\texttt{supershort\_rho\_bounds\_proposal} to get a set of \(\rho^{jk}\)
values from the ``supershort'' method in
Section~\ref{sec-selecting-baseline-plausible-bounds}.

With the bounds matrix and the data, the user can build a set of partial
identification bounds using \texttt{build\_bounds}, which can then be
displayed using \texttt{univariate\_bounds\_table}. The
\texttt{univariate\_bounds\_table} function also has a
\texttt{pin\_values} option allowing the user to fix the effect in a
chosen setting at a given part of its range (implementing ``pinning''
from Definition \ref{def:pinning}), to see how this impacts the
estimates in other settings.

\section{The Returns to College Education over
Time}\label{sec-college-over-time}

In this section I apply the methods proposed in this paper to produce
partial identification bounds on the labor market returns to attending
college in different settings in the United States. I first evaluate how
the labor market returns change over time.

For this application I use data from the American Community Survey
(ACS), a large scale annual data set collected by the US Census which
has data on educational attainment, income, and a set of background
characteristics. Five-year data files for 2003, 2008, 2013, 2018, and
2023 were collected from IPUMS \citep{ruggles2025}.

Because there is no apparent quasiexperiment that would allow for a
sharp identification of the returns to college separately in each year
in this data set, I instead partially identify the effect using
selection on observables. I first limit the data to individuals who have
completed high school, are not currently in the military, and have some
nonzero level of income. Then, the treatment variable is having attended
college for at least one year (not necessarily graduating).

Identification of the effect uses selection on observables, where income
is regressed on college attendance and a set of control variables.
Control variables are race (white / black / Asian / other), age, an
indicator for being hispanic, sex, veteran status, state, and an
indicator for speaking English.\footnote{Person-level survey weights are
  included.} There is almost certainly residual omitted variable bias
afer accounting for this set of controls, due to the omission of factors
like personality, parental socioeconomic status, and academic ability,
among other things.

To get a baseline set of partial identification bounds, I first apply
\citet{cinelli2019} separately in each year. While there are important
omitted variables, I believe I have included some of the most important
sources of endogeneity, and so I assume that the set of omitted
variables is no more than half as strong, as sources of bias, as the set
of included variables.\footnote{This section is intended as a
  demonstration of methods. The plausibility of any of the specific
  choices made for the purposes of partial identification could be
  debated.} This generates a set of ``original'' partial identification
bounds for the effect in each year.

Next, following the ``supershort'' method in
Section~\ref{sec-selecting-baseline-plausible-bounds}, I estimate the
effect of college attendance on income separately in each year with and
without any controls, calculate the change in the estimate as a result
of the removal of bias, and then calculate a matrix of proportionally
symmetric bounds on \(\rho^{jk}\) by seeing how the change in estimates
is related across settings. The table is shown in
Table~\ref{tbl-bounds-matrix}. Using these bounds, the partially
identified set for each year shrinks, since some effects in some years
are incompatible with possible values in other years, given the
\(\rho^{jk}\) bounds. I refer to these as ``univariate'' bounds.

\begin{table}

\caption{\label{tbl-bounds-matrix}Matrix of Bounds}

\centering{

 \centering \renewcommand*{\arraystretch}{1.1}
\begin{tabular}{llllll}
\hline
\hline
Year & 2003 & 2008 & 2013 & 2018 & 2023 \\ 
\hline
2003 &  & 1.312 & 1.008 & 1.210 & 1.557 \\ 
2008 & 0.762 &  & 1.301 & 1.586 & 2.042 \\ 
2013 & 0.992 & 0.769 &  & 1.219 & 1.569 \\ 
2018 & 0.827 & 0.630 & 0.820 &  & 1.287 \\ 
2023 & 0.642 & 0.490 & 0.637 & 0.777 & \\ 
\hline
\hline
\multicolumn{6}{l}{Upper-triangle and lower-triangle values show estimated upper and lower bounds on the proportional bias relationship.}\\ 
\end{tabular}

}

\end{table}%

The original and univariate bounds are shown in
Table~\ref{tbl-pib-year}. Notably, the addition of bounds between years
does not shrink the marginal partial identification bounds by much. Only
2013 has its range reduced. This result is consistent with
Section~\ref{sec-apx-marginal} whereby marginal sharpening only occurs
under certain conditions. The supershort \(\rho^{jk}\) bounds shown in
Table~\ref{tbl-bounds-matrix} are also fairly wide, indicating that
observed confounders resolve bias in similarly proportioned but not
tightly correlated ways across years. As such each part of the
identified range for most years is consistent with \emph{some} part of
the range of other years.

However, the impact of the bounds can be shown in the way that they
restrict the effects in different years to move together.
Figure~\ref{fig-pid-observed-bias-year} shows, on the left, the original
and univariate bounds. On the right I fix the 2003 estimate at the top
and bottom of its partial identification range, ``pinning'' it to those
values (implementing Definition \ref{def:pinning}), and show the
conditional identified ranges of the effect in other years. Here we see
a considerable reduction of plausible ranges for the estimates. These
estimates permit a wide range of partially identified effects in
general, but the joint restrictions show that, wherever the true value
is in that range, the assumptions made so far restrict the estimates in
each year to be at roughly similar levels, and identifies ranges that
show how the effect evolves over time.

\begin{table}

\caption{\label{tbl-pib-year}Partial Identification Bounds by Year}

\centering{

 \centering \renewcommand*{\arraystretch}{1.1}\resizebox{\textwidth}{!}{
\begin{tabular}{lp{.14\textwidth}p{.14\textwidth}p{.14\textwidth}p{.14\textwidth}p{.14\textwidth}}
\hline
\hline
Year & Estimate & New Lower Bound & New Upper Bound & Original Lower Bound & Original Upper Bound \\ 
\hline
2003 & 0.403 & 0.335 & 0.471 & 0.335 & 0.471 \\ 
2008 & 0.401 & 0.332 & 0.470 & 0.332 & 0.470 \\ 
2013 & 0.445 & 0.376 & 0.513 & 0.368 & 0.521 \\ 
2018 & 0.439 & 0.367 & 0.512 & 0.367 & 0.512 \\ 
2023 & 0.459 & 0.386 & 0.531 & 0.386 & 0.531\\ 
\hline
\hline
\end{tabular}
}

}

\end{table}%

\begin{figure}

\centering{

\pandocbounded{\includegraphics[keepaspectratio]{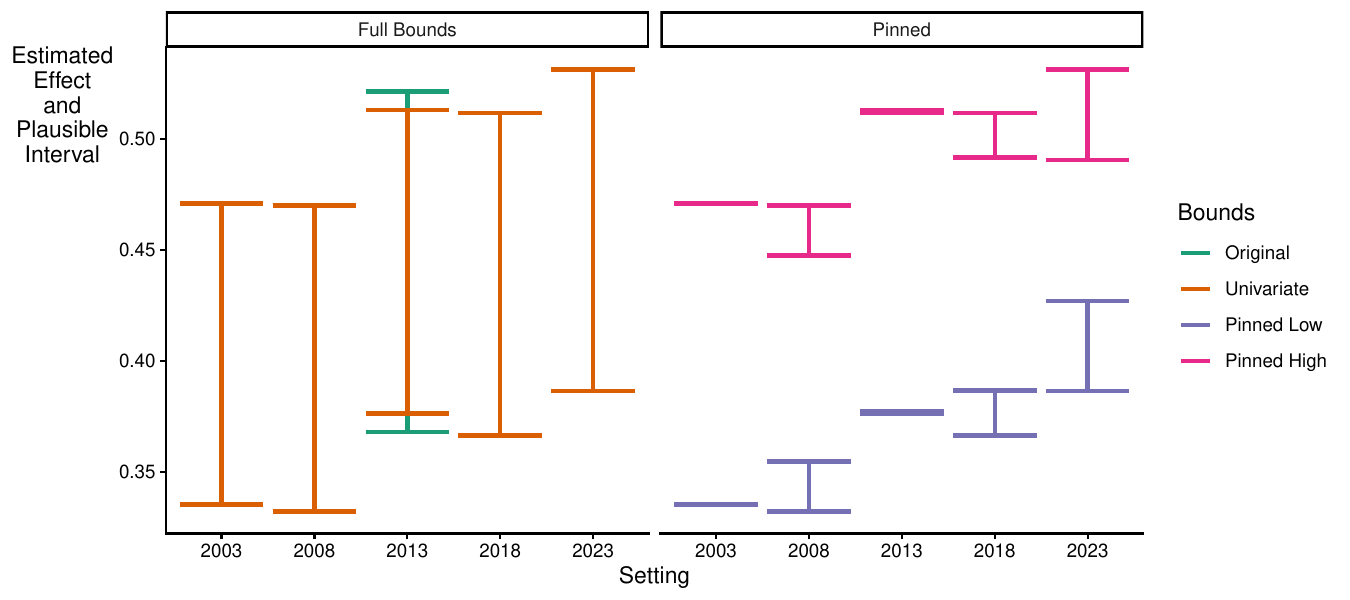}}

}

\caption{\label{fig-pid-observed-bias-year}Partial Identification Bounds
Using Observed Bias Relationships}

\end{figure}%

Figure~\ref{fig-pid-preset-year} repeats the analysis of
Figure~\ref{fig-pid-observed-bias-year}, but instead of constructing the
matrix of \(\rho^{jk}\) values using the supershort method and
observational data, it imposes a time-based maximum distance.
\(\rho^{jk}\) is set to be in the range \([.95,1/.95]\) for adjacent
5-year blocks, \([.95^2,1/(.95^2)]\) for effects that are two five-year
blocks away, and so on. The results are fairly similar, although
slightly more restrictive, given that the resulting bounds permit
smaller differences than under the supershort method, in this particular
application.

\begin{figure}

\centering{

\pandocbounded{\includegraphics[keepaspectratio]{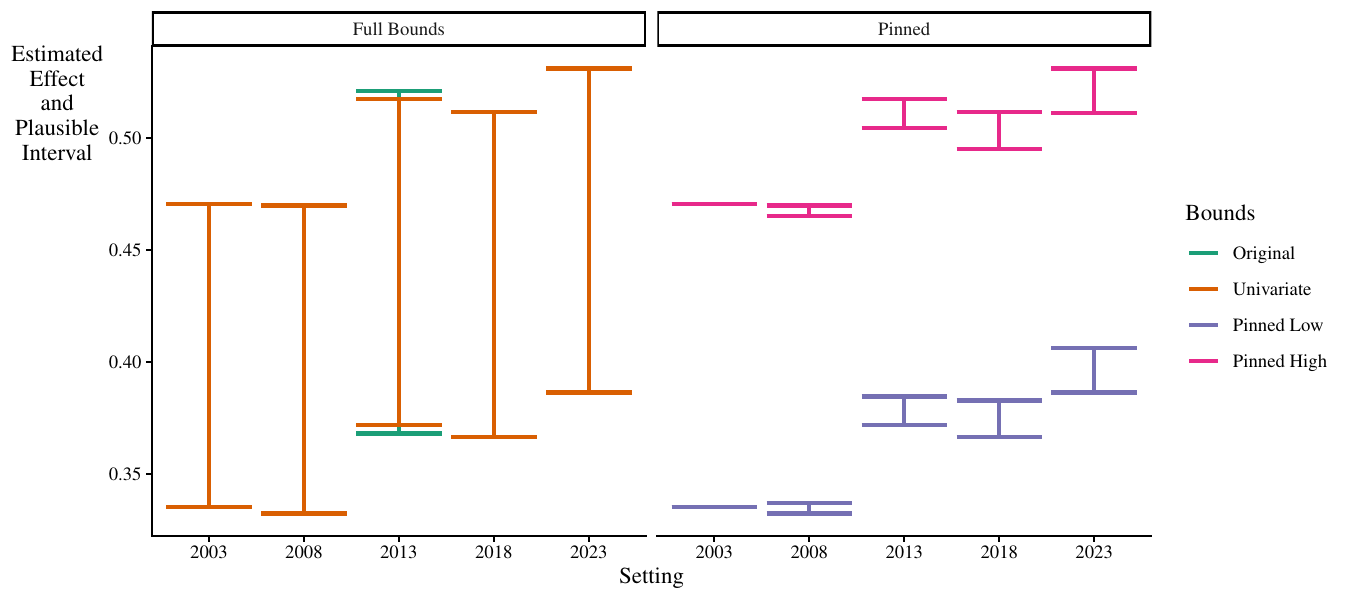}}

}

\caption{\label{fig-pid-preset-year}Partial Identification Bounds Using
Pre-Set .95 One-Year Bounds Restriction}

\end{figure}%

\section{The Returns to College Education by
Geography}\label{the-returns-to-college-education-by-geography}

In this section we will apply the same methodology as in the previous
section, but look at heterogeneity of effects over geography rather than
over time. In this context, the plausible assumptions on the bounds
between contexts is different.

Similar to in the previous section, I first construct a matrix of
proportionally symmetric \(\rho^{jk}\) bounds based on the
``supershort'' method. However, I only apply this restriction to
geographically adjacent states that share a border, setting no bound on
the bias relationship between non-adjacent states. Then, I impose an
adjacent-state set of restrictions, setting all \(\rho^{jk}\) values to
1.1, but only for adjacent states. States that are geographically
adjacent are restricted to \(\rho^{jk}=1.1\), and non-adjacent states
are set to have no restriction between
them.\footnote{This is just a demonstration of method; there are likely
  better assumptions to use here, either using the supershort method to
  get an average \(\rho^{jk}\) between adjacent states, or judging state
  similarity on a set of factors beyond geographic adjacency, also
  considering things like demographics.}\textsuperscript{,}\footnote{English-speaking
  status has been dropped as a covariate since it is perfectly collinear
  with other predictors in some of the states.}

Results are shown in Figure~\ref{fig-pid-state}. Like in
Section~\ref{sec-college-over-time}, the supershort method suggests
bounds that are too wide to impose any additional restrictions on the
partial identification bounds. However, the \(\rho^{jk}=1.1\) for
adjacent states restriction does meaningfully restrict partial
identification bounds. Notably, the bounds are narrowed most sharply for
states with the widest partial identification bounds on their estimates;
between-state restrictions add information and narrow bounds in these
cases where a state's own data leaves a wide range.

\begin{figure}

\centering{

\pandocbounded{\includegraphics[keepaspectratio]{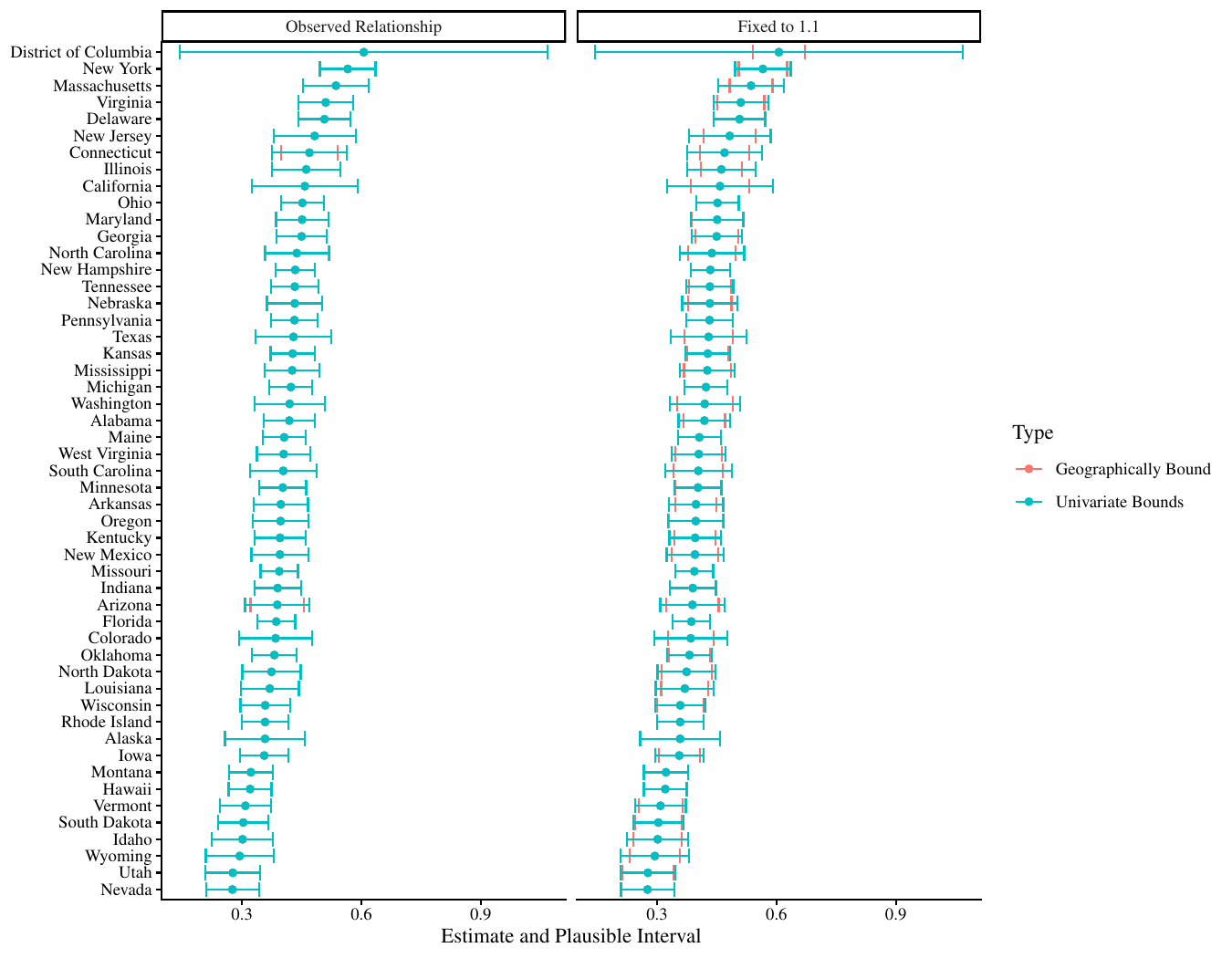}}

}

\caption{\label{fig-pid-state}Partial Identification Bounds for the
Effect of Education by State}

\end{figure}%

\section{Conclusion}\label{conclusion}

This paper proposes a method that allows researchers to restrict the
joint identified set for causal effect estimates across many different
settings, using the dependence structure between omitted variable biases
across settings. In many contexts, identification must often rely on a
conditional independence assumption that is unlikely to hold exactly,
and so partial identification is a more plausible approach to effect
estimation. The method in this paper allows researchers using partial
identification across settings to take advantage of the relationships
between those settings to narrow the joint identified set and provide a
more informative result. The paper's primary formal contribution,
established in Propositions \ref{prop:joint} through
\ref{prop:k-settings}, is a restriction on the joint identified polytope
\(\mathcal{J}\) for the vector of setting-specific effects.

This kind of approach is valuable because in many contexts, we are
interested in the heterogeneous effect of a treatment or policy, but do
not have a quasiexperimental design that applies in more than one
setting, or at least does not apply in all settings. Requiring only
quasiexperimental results restricts researchers to areas where we happen
to get lucky and find exogenous variation, and in multi-setting contexts
requires that that same exogenous variation functions repeatedly.

Partial identification with cross-setting restrictions is a principled
middle ground. Simply using control variables relies on an arguably
stronger assumption than making use of a quasiexperiment, but can also
be implemented in a much wider range of settings. When the cross-setting
proportionality assumptions are well-motivated, this is likely to lead
to a more representative and informative body of causal inference
results about effect heterogeneity.

Marginal identified sets for individual settings are not always
sharpened. Proposition \ref{prop:marginal} and Corollary
\ref{cor:symmetric} provide the formal conditions under which marginal
sharpening occurs. However, the joint set can be meaningfully sharpened
even when the marginal sets cannot. Then, the researcher can explore the
joint set. The ``pinning'' operation (Definition \ref{def:pinning})
leverages this joint restriction to answer practically important
questions about conditional heterogeneity: if the effect in one setting
is near the upper end of its identified range, what can be said about
other settings?

An important limitation of the current paper is the absence of
statistical inference. The identified sets \(\mathcal{I}^i\) and
\(\mathcal{J}\) are population objects and do not account for sampling
variability in \(\hat\theta^i_s\) or in the estimated bias bounds.
Extending the frameworks of \citet{imbens2004confidence} and
\citet{stoye2009more} to the multi-setting joint polytope setting is a
natural and important direction for future work.

\section{Statement on the Use of AI
Tools}\label{statement-on-the-use-of-ai-tools}

Assistance in coding and software development as well as improvements to
the mathematical formalization come from Claude Code using Opus 4.6. All
concepts were human-created, and a fully human-written draft from before
Claude made suggestions to improve the mathematical formalization is
available from the author.

\section*{References}\label{references}
\addcontentsline{toc}{section}{References}

\renewcommand{\bibsection}{}
\bibliography{References.bib}

\FloatBarrier
\newpage
\appendix
\renewcommand{\thesection}{\Alph{section}}
\setcounter{section}{0}
\counterwithin{figure}{section}
\counterwithin{table}{section}
\renewcommand{\section}[1]{\refstepcounter{section}%
\begin{Large}\bfseries Appendix \thesection: #1\end{Large}}
\renewcommand{\subsection}[1]{\refstepcounter{subsection}%
\begin{large}\bfseries \thesubsection: #1\end{large}}

\section{Marginal Sharpening Conditions}\label{sec-apx-marginal}

\begin{proof}
A value $\theta^j$ belongs to $\mathrm{proj}_j(\mathcal{J}^{jk})$ if and only if there exists $\theta^k \in \mathcal{I}^k$ such that $\theta^j - \theta^k \in \mathcal{I}^D_{jk} = [d_l, d_u]$, where $d_l = (\theta^j_s - \theta^k_s) - c^{jk}_u$ and $d_u = (\theta^j_s - \theta^k_s) - c^{jk}_l$.

The condition $\theta^j - \theta^k \in [d_l, d_u]$ is equivalent to $\theta^k \in [\theta^j - d_u,\; \theta^j - d_l]$. The intersection $[\theta^j - d_u, \theta^j - d_l] \cap \mathcal{I}^k$ is non-empty if and only if
\[
\theta^j - d_u \leq \theta^k_s - \nu^k_l \text{ and } \theta^k_s - \nu^k_u \leq \theta^j - d_l.
\]
Substituting $d_u = (\theta^j_s - \theta^k_s) - c^{jk}_l$ and $d_l = (\theta^j_s - \theta^k_s) - c^{jk}_u$ and simplifying:
\begin{align*}
\theta^j - d_u &= \theta^j - \theta^j_s + \theta^k_s + c^{jk}_l \leq \theta^k_s - \nu^k_l \\
&\Longleftrightarrow \theta^j \leq \theta^j_s - \nu^k_l - c^{jk}_l, \\[6pt]
\theta^k_s - \nu^k_u &\leq \theta^j - d_l = \theta^j - \theta^j_s + \theta^k_s + c^{jk}_u \\
&\Longleftrightarrow \theta^j \geq \theta^j_s - \nu^k_u - c^{jk}_u.
\end{align*}
Therefore $\mathrm{proj}_j(\mathcal{J}^{jk}) = \mathcal{I}^j \cap [\theta^j_s - \nu^k_u - c^{jk}_u,\; \theta^j_s - \nu^k_l - c^{jk}_l]$.

Strict sharpening of the lower bound occurs when $\theta^j_s - \nu^k_u - c^{jk}_u > \theta^j_s - \nu^j_u \Rightarrow \nu^j_u > \nu^k_u + c^{jk}_u$. Strict sharpening of the upper bound occurs when $\theta^j_s - \nu^k_l - c^{jk}_l < \theta^j_s - \nu^j_l \Rightarrow \nu^j_l < \nu^k_l + c^{jk}_l$. 

Intuitively, this says that marginal sharpening occurs whenever there is a parameter value in one setting that is farther away from the other setting's identified set than the widest possible difference between settings will allow.
\end{proof}

\begin{corollary}[No Marginal Sharpening with Equal Symmetric Bounds]
\label{cor:symmetric}
Suppose $\nu^j_l = -\nu^j_u = -\nu^{j*} < 0$ and $\nu^k_l = -\nu^k_u = -\nu^{k*} < 0$ (symmetric bias bounds) and also Assumption \ref{def:Csymmetry} applies such that the proportional bounds are symmetric and $\rho^{jk}_l=1/\rho^{jk}_u$. Then, $c^{jk}_u = (\rho^{jk}_u - 1)\nu^{k*}$ and $c^{jk}_l = -(\rho^{jk}_u - 1)\nu^{k*}$. The conditions for marginal sharpening (Proposition \ref{prop:marginal}) reduce to $\nu^{j*} > \rho^{jk}_u \nu^{k*}$. In particular, when both settings have equal bias bound magnitudes ($\nu^{j*} = \nu^{k*}$), strict marginal sharpening never occurs for any $\rho^{jk}_u > 1$.
\end{corollary}

\begin{proof}
With $\nu^k_l = -\nu^{k*}$, $\nu^k_u = \nu^{k*}$, and $\rho^{jk}_l=1/\rho^{jk}_u$, the elements of $\mathbb{C}^{jk}$ are
$\{-\frac{1-\rho^{jk}_u}{\rho^{jk}_u}\nu^{k*}, (1-\rho^{jk}_u)\nu^{k*}, \frac{1-\rho^{jk}_u}{\rho^{jk}_u}\nu^{k*}, (\rho^{jk}_u-1)\nu^{k*}\}$.
Since $\rho^{jk}_u > 1$ and $\bar{\nu}^k > 0$, the maximum is $c^{jk}_u = (\rho^{jk}_u-1)\nu^{k*}$ and the minimum is $c^{jk}_l = (1-\rho^{jk}_u) \nu^{k*}$. 

The lower-bound sharpening condition becomes $ \nu^{j*} >  \nu^{k*} + (\rho^{jk}_u - 1) \nu^{k*} = \rho^{jk}_u  \nu^{k*}$, which fails when $ \nu^{j*} =  \nu^{k*}$ since this implies $1 > \rho^{jk}_u$, when by assumption we have $\rho^{jk}_u > 1$. The upper-bound sharpening condition becomes $\nu^{j*} < \nu^{k*} + (1-\rho^{jk}_u) \nu^{k*}=\nu^{k*}(2-\rho^{jk}_u)$. This similarly fails when $ \nu^{j*} =  \nu^{k*}$ since this also implies $\rho^{jk}_u < 1$.

\end{proof}

\end{document}